\documentclass[letterpaper, 10 pt, conference]{ieeeconf}  

\IEEEoverridecommandlockouts                              
\usepackage{graphics} 
\usepackage{epsfig} 
\usepackage{algorithm,algorithmic}

\usepackage{balance}
\usepackage{amsmath,amsfonts,amssymb,amscd}
\usepackage{capt-of}
\usepackage{bbm}
\usepackage{color}
\usepackage{cite}
\usepackage{verbatim}

\newtheorem{remark}{\bfseries Remark}

\newtheorem{example}{\bfseries Example}

\newtheorem{prop}{\bfseries Proposition}

\newenvironment{varalgorithm}[1]
  {\algorithm\renewcommand{\thealgorithm}{#1}}
  {\endalgorithm}

\begin{document}

\title{\LARGE \bf
Distributed Average Consensus under Quantized Communication \\ via Event-Triggered Mass Summation}

\author{Apostolos~I.~Rikos 
and Christoforos~N.~Hadjicostis 
	\thanks{The authors are with the Department of Electrical and Computer Engineering at the University of Cyprus, Nicosia, Cyprus. E-mails:{\tt~\{arikos01,chadjic\}@ucy.ac.cy}.}
}
\maketitle
\thispagestyle{empty}
\pagestyle{empty}

\begin{abstract}

We study distributed average consensus problems in multi-agent systems with directed communication links that are subject to quantized information flow. The goal of distributed average consensus is for the nodes, each associated with some initial value, to obtain the average (or some value close to the average) of these initial values. In this paper, we present and analyze a distributed averaging algorithm which operates exclusively with quantized values (specifically, the information stored, processed and exchanged between neighboring agents is subject to deterministic uniform quantization) and relies on event-driven updates (e.g., to reduce energy consumption, communication bandwidth, network congestion, and/or processor usage). We characterize the properties of the proposed distributed averaging protocol on quantized values and show that its execution, on any time-invariant and strongly connected digraph, will allow all agents to reach, in finite time, a common consensus value represented as the ratio of two integer that is equal to the exact average. We conclude with examples that illustrate the operation, performance, and potential advantages of the proposed algorithm.

\end{abstract}


%
%
%
%
\section{INTRODUCTION}\label{intro}

In recent years, there has been a growing interest for control and coordination of networks consisting of multiple agents, like groups of sensors \cite{2005:XiaoBoydLall} or mobile autonomous agents \cite{2004:Murray}. 
A problem of particular interest in distributed control is the \textit{consensus} problem where the objective is to develop distributed algorithms that can be used by a group of agents in order to reach agreement to a common decision. 
The agents start with different initial values/information and are allowed to communicate locally via inter-agent information exchange under some constraints on connectivity. 
Consensus processes play an important role in many problems, such as leader election \cite{1996:Lynch},  motion coordination of multi-vehicle systems \cite{2005:Olshevsky_Tsitsiklis, 2004:Murray}, and clock synchronization \cite{2007:Gamba}.


One special case of the consensus problem is distributed averaging, where each agent (initially endowed with a numerical value) can send/receive information to/from other agents in its neighborhood and update its value iteratively, so that eventually, it is able to compute the average of all initial values. 
Average consensus is an important problem \cite{2018:BOOK_Hadj, 2005:Olshevsky_Tsitsiklis, 2008:Sundaram_Hadjicostis, 2013:Themis_Hadj_Johansson, 2004:XiaoBoyd, 2010:Dimakis_Rabbat, 2011:Morse_Yu, 1984:Tsitsiklis} and has been studied extensively in settings where each agent processes and transmits real-valued states with infinite precision. 

More recently, researchers have also studied the case when network links can only allow messages of limited length to be transmitted between agents (presumably due to constraints on their capacity), effectively extending techniques for average consensus towards the direction of quantized consensus.
Various probabilistic strategies have been proposed, allowing the agents in a network to reach quantized consensus with probability one \cite{2007:Aysal_Rabbat, 2012:Lavaei_Murray, 2007:Basar, 2009:Carli_Zampieri, 2016:Chamie_Basar, 2011:Cai_Ishii}. 
Furthermore, in many types of communication networks it is desirable to update values infrequently to avoid consuming valuable network resources. 
Thus, there is an increasing need for novel event-triggered algorithms for cooperative control, which aim at more efficient usage of network resources \cite{2013:Dimarogonas_Johansson, 2014:nowzari_cortes, 2012:Liu_Chen}.



In this paper, we present a novel distributed average consensus algorithm that combines the both of the features mentioned above. More specifically, the processing, storing, and exchange of information between neighboring agents is ``event-driven'' and  subject to uniform quantization. 
Following \cite{2007:Basar, 2011:Cai_Ishii} we assume that the states are integer-valued (which comprises a class of quantization effects). 
We note that most work dealing with quantization has concentrated on the scenario where the agents have real-valued states but can transmit only quantized values through limited rate channels (see, e.g., \cite{2008:Carli_Zampieri, 2016:Chamie_Basar}). 
By contrast, our assumption is also suited to the case where the states are stored in digital memories of finite capacity (as in \cite{2009:Nedic, 2007:Basar, 2011:Cai_Ishii}) and the control actuation of each node is event-based, 
which enables more efficient use of available resources. 
The main result of this paper shows that the proposed algorithm will allow all agents to reach quantized consensus in finite time by reaching a value represented as the ratio of two integer values that is equal to the average.

%
%
%
%
\section{PRELIMINARIES}\label{notation}

The sets of real, rational, integer and natural numbers are denoted by $ \mathbb{R}, \mathbb{Q}, \mathbb{Z}$ and $\mathbb{N}$, respectively. 
The symbol $\mathbb{Z}_+$ denotes the set of nonnegative integers.

Consider a network of $n$ ($n \geq 2$) agents communicating only with their immediate neighbors. 
The communication topology can be captured by a directed graph (digraph), called \textit{communication digraph}. 
A digraph is defined as $\mathcal{G}_d = (\mathcal{V}, \mathcal{E})$, where $\mathcal{V} =  \{v_1, v_2, \dots, v_n\}$ is the set of nodes and $\mathcal{E} \subseteq \mathcal{V} \times \mathcal{V} - \{ (v_j, v_j) \ | \ v_j \in \mathcal{V} \}$ is the set of edges (self-edges excluded). 
A directed edge from node $v_i$ to node $v_j$ is denoted by $m_{ji} \triangleq (v_j, v_i) \in \mathcal{E}$, and captures the fact that node $v_j$ can receive information from node $v_i$ (but not the other way around). 
We assume that the given digraph $\mathcal{G}_d = (\mathcal{V}, \mathcal{E})$ is \textit{static} (i.e., does not change over time) and \textit{strongly connected} (i.e., for each pair of nodes $v_j, v_i \in \mathcal{V}$, $v_j \neq v_i$, there exists a directed \textit{path} from $v_i$ to $v_j$). 
The subset of nodes that can directly transmit information to node $v_j$ is called the set of in-neighbors of $v_j$ and is represented by $\mathcal{N}_j^- = \{ v_i \in \mathcal{V} \; | \; (v_j,v_i)\in \mathcal{E}\}$, while the subset of nodes that can directly receive information from node $v_j$ is called the set of out-neighbors of $v_j$ and is represented by $\mathcal{N}_j^+ = \{ v_l \in \mathcal{V} \; | \; (v_l,v_j)\in \mathcal{E}\}$. 
The cardinality of $\mathcal{N}_j^-$ is called the \textit{in-degree} of $v_j$ and is denoted by $\mathcal{D}_j^-$ (i.e., $\mathcal{D}_j^- = | \mathcal{N}_j^- |$), while the cardinality of $\mathcal{N}_j^+$ is called the \textit{out-degree} of $v_j$ and is denoted by $\mathcal{D}_j^+$ (i.e., $\mathcal{D}_j^+ = | \mathcal{N}_j^+ |$).

We assume that each node is aware of its out-neighbors and can directly (or indirectly\footnote{Indirect transmission could involve broadcasting a message to all out-neighbors while including in the message header the ID of the out-neighbor it is intended for.}) transmit messages to each out-neighbor; however, it cannot necessarily receive messages from them. 
In the randomized version of the protocol, each node $v_j$ assigns a nonzero \textit{probability} $b_{lj}$ to each of its outgoing edges $m_{lj}$ (including a virtual self-edge), where $v_l \in \mathcal{N}^+_j \cup \{ v_j \}$. 
This probability assignment can be captured by a column stochastic matrix $\mathcal{B} = [b_{lj}]$. 
A very simple choice would be to set  
\begin{align*}
b_{lj} = \left\{ \begin{array}{ll}
         \frac{1}{1 + \mathcal{D}_j^+}, & \mbox{if $v_{l} \in \mathcal{N}_j^+ \cup \{v_j\}$,}\\
         0, & \mbox{otherwise.}\end{array} \right. 
\end{align*}
Each nonzero entry $b_{lj}$ of matrix $\mathcal{B}$ represents the probability of node $v_j$ transmitting towards the out-neighbor $v_l \in \mathcal{N}^+_j$ through the edge $m_{lj}$, or performing no transmission\footnote{From the definition of $\mathcal{B} = [b_{lj}]$ we have that $b_{jj} = \frac{1}{1 + \mathcal{D}_j^+}$, $\forall v_j \in \mathcal{V}$. This represents the probability that node $v_j$ will not perform a transmission to any of its out-neighbors $v_l \in \mathcal{N}^+_j$ (i.e., it will transmit to itself).}.

In the deterministic version of the protocol, each node $v_j$ also assigns a \textit{unique order} in the set $\{0,1,..., \mathcal{D}_j^+ -1\}$ to each of its outgoing edges $m_{lj}$, where $v_l \in \mathcal{N}^+_j$. 
The order of link $(v_l,v_j)$ for node $v_j$ is denoted by $P_{lj}$ (such that $\{P_{lj} \; | \; v_l \in \mathcal{N}^+_j\} = \{0,1,..., \mathcal{D}_j^+ -1\}$). 
This unique predetermined order is used during the execution of the proposed distributed algorithm as a way of allowing node $v_j$ to transmit messages to its out-neighbors in a \textit{round-robin}\footnote{When executing the deterministic protocol, each node $v_j$ transmits to its out-neighbors by following a predetermined order. The next time it needs to transmit to an out-neighbor, it will continue from the outgoing edge it stopped the previous time and cycle through the edges in a round-robin fashion according to the predetermined ordering.} fashion.

%
%
%
%
\section{PROBLEM FORMULATION}\label{probForm}

Consider a strongly connected digraph $\mathcal{G}_d = (\mathcal{V}, \mathcal{E})$, where each node $v_j \in \mathcal{V}$ has an initial (i.e., for $k=0$) quantized value $y_j[0]$ (for simplicity, we take $y_j[0] \in \mathbb{Z}$). 
In this paper, we develop a distributed algorithm that allows nodes (while processing and transmitting \textit{quantized} information via available communication links between nodes) to eventually obtain, after a finite number of steps, a quantized fraction $q^s$ which is equal to the average $q$ of the initial values of the nodes, where
\begin{equation}
q = \frac{\sum_{l=1}^{n}{y_l[0]}}{n} .
\end{equation}

\begin{remark}
Following \cite{2007:Basar, 2011:Cai_Ishii} we assume that the state of each node is integer valued. 
This abstraction subsumes a class of quantization effects (e.g., uniform quantization).
\end{remark}

The quantized average $q^s$ is defined as the ceiling $q^s = \lceil q \rceil$ or the floor $q^s = \lfloor q \rfloor$ of the true average $q$ of the initial values. 
Let $S \triangleq \mathbf{1}^{\rm T} y[0]$, where $\mathbf{1} = [1 \ ... \ 1]^{\rm T}$ is the vector of ones, and let $y[0] = [y_1[0] \ ... \ y_n[0]]^{\rm T}$ be the vector of the quantized initial values. 
We can write $S$ uniquely as $S = nL + R$ where $L$ and $R$ are both integers and $0 \leq R < n$. 
Thus, we have that either $L$ or $L+1$ may be viewed as an integer approximation of the average of the initial values $S/n$ (which may not be integer in general).

The algorithm we will develop will be iterative. 
With respect to quantization of information flow, we have that at time step $k \in \mathbb{Z}_+$ (where $\mathbb{Z}_+$ is the set of nonnegative integers), each node $v_j \in \mathcal{V}$ maintains the state variables $y^s_j, z^s_j, q_j^s$, where $y^s_j \in \mathbb{Z}$, $z^s_j \in \mathbb{N}$ and  $q_j^s$ (where $q_j^s = \frac{y_j^s}{z_j^s}$), and the mass variables $y_j, z_j$ where $y_j \in \mathbb{Z}$ and $z_j \in \mathbb{N}_0$.
The aggregate states are denoted by $y^s[k] = [y^s_1[k] \ ... \ y^s_n[k]]^{\rm T} \in \mathbb{Z}^n$, $z^s[k] = [z^s_1[k] \ ... \ z^s_n[k]]^{\rm T} \in \mathbb{N}^n$, $q^s[k] = [q^s_1[k] \ ... \ q^s_n[k]]^{\rm T} \in \mathbb{Q}^n$ and $y[k] = [y_1[k] \ ... \ y_n[k]]^{\rm T} \in \mathbb{Z}^n$, $z[k] = [z_1[k] \ ... \ z_n[k]]^{\rm T} \in \mathbb{N}^n$ respectively. 

Following the execution of the proposed distributed algorithm, we argue that $\exists \ k_0$ so that for every $k \geq k_0$ we have 
\begin{equation}\label{alpha_z_y}
y^s_j[k] = \frac{\sum_{l=1}^{n}{y_l[0]}}{\alpha}  \ \ \text{and} \ \ z^s_j[k] = \frac{n}{\alpha} ,
\end{equation}
where $\alpha \in \mathbb{N}$. This means that 
\begin{equation}\label{alpha_q}
q^s_j[k] = \frac{(\sum_{l=1}^{n}{y_l[0]}) / \alpha}{n / \alpha} = q ,
\end{equation}
for every $v_j \in \mathcal{V}$ (i.e., for $k \geq k_0$ every node $v_j$ has calculated $q$ as the ratio of two integer values).


%
%
%
%
\section{RANDOMIZED QUANTIZED AVERAGING ALGORITHM}\label{algorithm}

In this section we propose a distributed information exchange process in which the nodes transmit and receive quantized messages so that they reach quantized average consensus on their initial values after a finite number of steps.

\noindent
The operation of the proposed distributed algorithm is summarized below.

\noindent
\textbf{Initialization:}
Each node $v_j$ selects a set of probabilities $\{ b_{lj} \ | \ v_{l} \in \mathcal{N}_j^+ \cup \{v_j\} \}$ such that $0 < b_{lj} < 1$ and $\sum_{v_{l} \in \mathcal{N}_j^+ \cup \{v_j\}} b_{lj} = 1$ (see Section~\ref{notation}). 
Each value $b_{lj}$, represents the probability for node $v_j$ to transmit towards out-neighbor $v_l \in \mathcal{N}^+_j$ (or perform no transmission), at any given time step (independently between time steps).  
Each node has some initial value $y_j[0]$, and also sets its state variables, for time step $k=0$, as $z_j[0] = 1$, $z^s_j[0] = 1$ and $y^s_j[0] = y_j[0]$, which means that $q^s_j[0] = y_j[0] / 1$. 

\noindent
The iteration involves the following steps:

\noindent
\textbf{Step 1. Transmitting:} According to the nonzero probabilities $b_{lj}$, assigned by node $v_j$ during the initialization step, it either transmits $z_j[k]$ and $y_j[k]$ towards out-neighbor $v_l \in \mathcal{N}_j^+$ or performs no transmission. 
If it performs a transmission towards an out-neighbor $v_l \in \mathcal{N}_j^+$, it sets $y_j[k] = 0$ and $z_j[k] = 0$.

\noindent
\textbf{Step 2. Receiving:} Each node $v_j$ receives messages $y_i[k]$ and $z_i[k]$ from its in-neighbors $v_i \in \mathcal{N}_j^-$, and it sums them along with its stored messages $y_j[k]$ and $z_j[k]$ as
$$
y_j[k+1] = \sum_{v_i \in \mathcal{N}_j^- \cup \{v_j\}} w_{ji}[k] y_i[k] ,
$$
and 
$$
z_j[k+1] = \sum_{v_i \in \mathcal{N}_j^- \cup \{v_j\}} w_{ji}[k] z_i[k] ,
$$
where $w_{ji}[k] = 0$ if no message is received from in-neighbor $v_i \in \mathcal{N}_j^-$; otherwise $w_{ji}[k] = 1$.

\noindent
\textbf{Step 3. Processing:} 
If $z_j[k+1] \geq z^s_j[k]$, node $v_j$ sets $z^s_j[k+1] = z_j[k+1]$, $y^s_j[k+1] = y_j[k+1]$ and 
$$
q^s_j[k+1] = \frac{y^s_j[k+1]}{z^s_j[k+1]} .
$$
Then, $k$ is set to $k+1$ and the iteration repeats (it goes back to Step~1).


The probabilistic quantized mass transfer process is detailed as Algorithm~\ref{algorithm_prob} below (for the case when $b_{lj} = 1/(1+\mathcal{D}_j^+)$ for $v_l \in \mathcal{N}_j^+ \cup \{ v_j \}$ and $b_{lj}=0$ otherwise). 

\noindent
\vspace{-0.5cm}    
\begin{varalgorithm}{1}
\caption{Probabilistic Quantized Average Consensus}
\textbf{Input} 
\\ 1) A strongly connected digraph $\mathcal{G}_d = (\mathcal{V}, \mathcal{E})$ with $n=|\mathcal{V}|$ nodes and $m=|\mathcal{E}|$ edges. 
\\ 2) For every $v_j$ we have $y_j[0] \in \mathbb{Z}$. 
\\
\textbf{Initialization} 
\\ Every node $v_j \in \mathcal{V}$: 
\\ 1) Assigns a nonzero probability $b_{lj}$ to each of its outgoing edges $m_{lj}$, where $v_l \in \mathcal{N}^+_j$, as follows  
\begin{align*}
b_{lj} = \left\{ \begin{array}{ll}
         \frac{1}{1 + \mathcal{D}_j^+}, & \mbox{if $l = j$ or $v_{l} \in \mathcal{N}_j^+$,}\\
         0, & \mbox{if $l \neq j$ and $v_{l} \notin \mathcal{N}_j^+$.}\end{array} \right. 
\end{align*}
\\ 2) Sets $z_j[0] = 1$, $z^s_j[0] = 1$ and $y^s_j[0] = y_j[0]$ (which means that $q^s_j[0] = y_j[0] / 1$). 
\\
\textbf{Iteration}
\\ For $k=0,1,2,\dots$, each node $v_j \in \mathcal{V}$ does the following:
\\ 1) It either transmits $y_j[k]$ and $z_j[k]$ towards a randomly chosen out-neighbor $v_l \in \mathcal{N}_j^+$ (according to the nonzero probability $b_{lj}$) or performs no transmission (according to the nonzero probability $b_{jj}$). If it transmitted towards an out-neighbor, it sets $y_j[k] = 0$ and $z_j[k] = 0$.
\\ 2) It receives $y_i[k]$ and $z_i[k]$ from its in-neighbors $v_i \in \mathcal{N}_j^-$ and sets 
$$
y_j[k+1] = \sum_{v_i \in \mathcal{N}_j^- \cup \{v_j\}} w_{ji}[k] y_i[k] ,
$$
and 
$$
z_j[k+1] = \sum_{v_i \in \mathcal{N}_j^- \cup \{v_j\}} w_{ji}[k] z_i[k] ,
$$
where $w_{ji}[k] = 1$ if node $v_j$ receives values from node $v_i$ (otherwise $w_{ji}[k] = 0$). 
\\ 3) If the following condition holds,
\begin{equation}
z_j[k+1] \geq z^s_j[k],
\end{equation}
it sets $z^s_j[k+1] = z_j[k+1]$, $y^s_j[k+1] = y_j[k+1]$, which means that $q^s_j[k+1] = \frac{y^s_j[k+1]}{z^s_j[k+1]}$. 
\\ 4) It repeats (increases $k$ to $k + 1$ and goes back to Step~1).
\label{algorithm_prob}
\end{varalgorithm}


\begin{example}

Consider the strongly connected digraph $\mathcal{G}_d = (\mathcal{V}, \mathcal{E})$ shown in Fig.~\ref{prob_example}, with $\mathcal{V} = \{ v_1, v_2, v_3, v_4 \}$ and $\mathcal{E} = \{ m_{21}, m_{31}, m_{42}, m_{13}, m_{23}, m_{34} \}$, where each node has initial quantized values $y_1[0] = 5$, $y_2[0] = 3$, $y_3[0] = 7$, and $y_4[0] = 2$ respectively. 
The average $q$ of the initial values of the nodes, is equal to $q = \frac{17}{4}$.

\begin{figure}[h]
\begin{center}
\includegraphics[width=0.25\columnwidth]{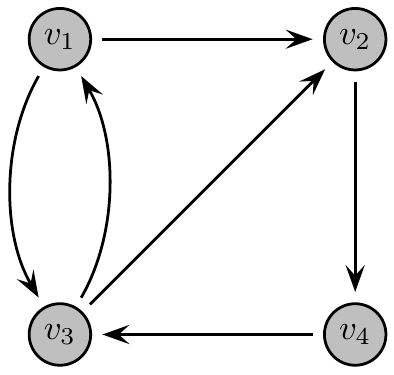}
\caption{Example of digraph for probabilistic quantized averaging.}
\label{prob_example}
\end{center}
\end{figure}

Each node $v_j \in \mathcal{V}$ follows the Initialization steps ($1-2$) in Algorithm~\ref{algorithm_prob}, assigning to each of its outgoing edges $v_l \in \mathcal{N}_j^+ \cup \{v_j\}$ a nonzero probability value $b_{lj}$ equal to $b_{lj} = \frac{1}{1 + \mathcal{D}_j^+}$. 
The assigned values can be seen in the following matrix
\[
\mathcal{B}
=
\begin{bmatrix}
    \frac{1}{3} & 0 & \frac{1}{3} &  0\\ \vspace{-0.35cm} \\
    \frac{1}{3} & \frac{1}{2} & \frac{1}{3} &  0\\ \vspace{-0.35cm} \\
    \frac{1}{3} & 0 & \frac{1}{3} &  \frac{1}{2}\\ \vspace{-0.35cm} \\
    0 & \frac{1}{2} & 0 &  \frac{1}{2}\\ \vspace{-0.35cm} \\
\end{bmatrix}
,
\]
while the initial mass and state variables are shown in Table~\ref{tableProb}. 

For the execution of the proposed algorithm, suppose that at time step $k=0$, nodes $v_1$, $v_3$ and $v_4$ transmit to nodes $v_2$, $v_1$ and $v_3$, respectively, whereas node $v_2$, performs no transmission. 
The mass and state variables for $k=1$ are shown in Table~\ref{tableProbk_1}. 

\begin{center}
\captionof{table}{Initial Mass and State Variables for Fig.~\ref{prob_example}}
\label{tableProb}
{\small 
\begin{tabular}{|c||c|c|c|c|c|}
\hline
Nodes &\multicolumn{5}{c|}{Mass and State Variables for $k=0$}\\
$v_j$ &$y_j[0]$&$z_j[0]$&$y^s_j[0]$&$z^s_j[0]$&$q^s_j[0]$\\
\cline{2-6}
 &  &  &  &  & \\
$v_1$ & 5 & 1 & 5 & 1 & 5 / 1\\
$v_2$ & 3 & 1 & 3 & 1 & 3 / 1\\
$v_3$ & 7 & 1 & 7 & 1 & 7 / 1\\
$v_4$ & 2 & 1 & 2 & 1 & 2 / 1\\
\hline
\end{tabular}
}
\end{center}
\vspace{0.4cm}

\begin{center}
\captionof{table}{Mass and State Variables for Fig.~\ref{prob_example} for $k=1$}
\label{tableProbk_1}
{\small 
\begin{tabular}{|c||c|c|c|c|c|}
\hline
Nodes &\multicolumn{5}{c|}{Mass and State Variables for $k=1$}\\
$v_j$ &$y_j[1]$&$z_j[1]$&$y^s_j[1]$&$z^s_j[1]$&$q^s_j[1]$\\
\cline{2-6}
 &  &  &  &  & \\
$v_1$ & 7 & 1 & 7 & 1 & 7 / 1\\
$v_2$ & 8 & 2 & 8 & 2 & 8 / 2\\
$v_3$ & 2 & 1 & 2 & 1 & 2 / 1\\
$v_4$ & 0 & 0 & 2 & 1 & 2 / 1\\
\hline
\end{tabular}
}
\end{center}
\vspace{0.2cm}

It is important to notice here that nodes $v_1$ and $v_3$ have mass variables $y_1[1] = y_3[0] = 7$, $z_1[1] = z_3[0] = 1$ and $y_3[1] = y_4[0] = 2$, $z_3[1] = z_4[0] = 1$ (and update their state variables), while node $v_2$ has mass variables $y_2[1] = y_1[0] + y_2[0] = 8$, $z_2[1] = z_1[0] + z_2[0] = 2$ (also updating its state variables). 
In the latter case we can say that the mass variables of nodes $v_1$ and $v_2$ will ``merge''. 

Suppose now that at time step $k=1$, nodes $v_1$ and $v_2$ transmit to nodes $v_3$ and $v_4$. 
Node $v_3$, does not perform a transmission while node $v_4$ has no mass to transmit. 
The mass and state variables for $k=2$ are shown in Table~\ref{tableProbk_2}.

\begin{center}
\captionof{table}{Mass and State Variables for Fig.~\ref{prob_example} for $k=2$} 
\label{tableProbk_2}
{\small 
\begin{tabular}{|c||c|c|c|c|c|}
\hline
Nodes &\multicolumn{5}{c|}{Mass and State Variables for $k=2$}\\
$v_j$ &$y_j[2]$&$z_j[2]$&$y^s_j[2]$&$z^s_j[2]$&$q^s_j[2]$\\
\cline{2-6}
 &  &  &  &  & \\
$v_1$ & 0 & 0 & 7 & 1 & 7 / 1\\
$v_2$ & 0 & 0 & 8 & 2 & 8 / 2\\
$v_3$ & 9 & 2 & 9 & 2 & 9 / 2\\
$v_4$ & 8 & 2 & 8 & 2 & 8 / 2\\
\hline
\end{tabular}
}
\end{center}
\vspace{0.2cm}

Then, suppose that at time step $k=2$, node $v_4$ transmits to node $v_3$, while node $v_3$, does not perform a transmission (nodes $v_1$ and $v_2$ have no mass to transmit). 
The mass and state variables for $k=3$ are shown in Table~\ref{tableProbk_3}.

We can see that, at time step $k=3$ all the initial mass variables are ``merged'' in node $v_3$ (i.e., $y_3[3] = y_1[0] + y_2[0] + y_3[0] + y_4[0]$ and $z_3[3] = z_1[0] + z_2[0] + z_3[0] + z_4[0]$). 
Now suppose that during time steps $k=3, 4, 5$ the following transmissions take place: ``$v_3$ transmits to $v_1$'', ``$v_1$ transmits to $v_2$'', ``$v_2$ transmits to $v_4$''. 
The mass and state variables for $k=5$ are shown in Table~\ref{tableProbk_5}.

\begin{center}
\captionof{table}{Mass and State Variables for Fig.~\ref{prob_example} for $k=3$}
\label{tableProbk_3}
{\small 
\begin{tabular}{|c||c|c|c|c|c|}
\hline
Nodes &\multicolumn{5}{c|}{Mass and State Variables for $k=3$}\\
$v_j$ &$y_j[3]$&$z_j[3]$&$y^s_j[3]$&$z^s_j[3]$&$q^s_j[3]$\\
\cline{2-6}
 &  &  &  &  & \\
$v_1$ & 0 & 0 & 7 & 1 & 7 / 1\\
$v_2$ & 0 & 0 & 8 & 2 & 8 / 2\\
$v_3$ & 17 & 4 & 17 & 4 & 17 / 4\\
$v_4$ & 0 & 0 & 8 & 2 & 8 / 2\\
\hline
\end{tabular}
}
\end{center}
\vspace{0.4cm}

\begin{center}
\captionof{table}{Mass and State Variables for Fig.~\ref{prob_example} for $k=5$}
\label{tableProbk_5}
{\small 
\begin{tabular}{|c||c|c|c|c|c|}
\hline
Nodes &\multicolumn{5}{c|}{Mass and State Variables for $k=3$}\\
$v_j$ &$y_j[5]$&$z_j[5]$&$y^s_j[5]$&$z^s_j[5]$&$q^s_j[5]$\\
\cline{2-6}
 &  &  &  &  & \\
$v_1$ & 0 & 0 & 17 & 4 & 17 / 4\\
$v_2$ & 0 & 0 & 17 & 4 & 17 / 4\\
$v_3$ & 0 & 0 & 17 & 4 & 17 / 4\\
$v_4$ & 17 & 4 & 17 & 4 & 17 / 4\\
\hline
\end{tabular}
}
\end{center}
\vspace{0.2cm}

From Table~\ref{tableProbk_5}, we can see that for $k \geq 5$ it holds that 
$$
q_j^s[k] = q = \frac{17}{4} ,
$$
for every $v_j \in \mathcal{V}$, which means that every node $v_j$ will eventually obtain a quantized fraction $q_j^s$, which is equal to the average $q$ of the initial values of the nodes. 
\end{example}

\begin{remark}
From the previous example, it is important to notice that, once the initial mass variables ``merge'' at time step $k = 3$, they \textit{remain} ``merged'' during the operation of Algorithm~\ref{algorithm_prob} for every time step $k \geq 3$. 
\end{remark}


We are now ready to prove that during the operation of Algorithm~\ref{algorithm_prob} each agent obtains two integer values $y^s$ and $z^s$, the ratio of which is equal to the average $q$ of the initial values of the nodes. 

\begin{prop}
\label{PROP1_prob}
Consider a strongly connected digraph $\mathcal{G}_d = (\mathcal{V}, \mathcal{E})$ with $n=|\mathcal{V}|$ nodes and $m=|\mathcal{E}|$ edges, and $z_j[0] = 1$ and $y_j[0] \in \mathbb{Z}$ for every node $v_j \in \mathcal{V}$ at time step $k=0$. 
Suppose that each node $v_j \in \mathcal{V}$ follows the Initialization and Iteration steps as described in Algorithm~\ref{algorithm_prob}. 
Let $\mathcal{V}^+[k] \subseteq \mathcal{V}$ be the set of nodes $v_j$ with positive mass variable $z_j[k]$ at iteration $k$ (i.e., $\mathcal{V}^+[k] = \{ v_j \in \mathcal{V} \; | \; z_j[k] > 0 \}$). 
During the execution of Algorithm~\ref{algorithm_prob}, for every $k \geq 0$, we have that
$$
1 \leq | \mathcal{V}^+[k+1] | \leq | \mathcal{V}^+[k] | \leq n .
$$
\end{prop}
 
\begin{proof}
During the Iteration Steps~$1$ and $2$ of Algorithm~\ref{algorithm_prob}, at time step $k$, we have that each node $v_j \in \mathcal{V}$ transmits $z_j[k]$ and $y_j[k]$ towards a randomly chosen out-neighbor $v_l \in \mathcal{N}_j^+$, or performs no transmission. 
Then, it receives $y_i[k]$ and $z_i[k]$ from its in-neighbors $v_i \in \mathcal{N}_j^-$ and sets $y_j[k+1] = \sum_{v_i \in \mathcal{N}_j^- \cup \{v_j\}} w_{ji}[k] y_i[k]$, and $z_j[k+1] = \sum_{v_i \in \mathcal{N}_j^- \cup \{v_j\}} w_{ji}[k] z_i[k]$. 
The Iteration Steps~$1$ and $2$ of Algorithm~\ref{algorithm_prob}, during time step $k$, can be expressed according to the following equations 
\begin{equation}\label{y_exch}
y[k+1] = \mathcal{W}[k] \ y[k] ,
\end{equation}
and 
\begin{equation}\label{z_exch}
z[k+1] = \mathcal{W}[k] \ z[k] ,
\end{equation}
where $y[k] = [y_1[k] \ ... \ y_n[k]]^{\rm T}$, $z[k] = [z_1[k] \ ... \ z_n[k]]^{\rm T}$ and $\mathcal{W}[k] = [w_{lj}[k]]$ is an $n \times n$ binary (i.e., for every $k$, $w_{lj}[k]$ is either equal to $1$ or $0$, for every $(v_l,v_j) \in \mathcal{E}$), column stochastic matrix. 

Focusing on (\ref{z_exch}), during time step $k_0$, let us assume without loss of generality that
$z[k_0] = \left [ z_1[k_0] \dots z_{p_0}[k_0] \; 0 \dots \; 0 \right ]^{\rm T}$, which means that we have $z_i[k_0] > 0$, $\forall \ v_i \in \{ v_1, \cdots, v_{p_0} \}$ and $z_l[k_0] = 0$, $\forall \ v_l \in \mathcal{V} - \{ v_1, \cdots, v_{p_0} \}$. 
We can assume without loss of generality that the nodes with zero mass do not transmit (transmit to themselves). 
Let us consider the scenario where $\sum_{v_i \in \mathcal{N}_j^- \cup \{v_j\}} w_{ji}[k_0] = 1$, $\forall \ v_j \in \mathcal{V}$ (i.e., for every row of $\mathcal{W}[k_0]$ \textit{exactly} one element is equal to $1$ and all the other are equal to zero). 
This means that every node $v_j$ will receive exactly one mass variable $z_i[k_0]$ (the bottom $n-p_0$ nodes receive their own mass). 
Since, at time step $k_0$, we have $p_0$ nodes with nonzero mass variables, we have that at time step $k_0 + 1$, exactly $p_0$ nodes have a nonzero mass variable. 
As a result, for this scenario, we have $| \mathcal{V}^+[k_0+1] | = | \mathcal{V}^+[k_0] |$. 

Without loss of generality, let us consider the scenario where $w_{ji_1}[k_0] = 1$, $w_{ji_2}[k_0] = 1$ (where $v_{i_1}, v_{i_2} \ \in \mathcal{N}_j^- \cup \{ v_j \}$) and $w_{ji}[k_0] = 0, \forall \ v_i \in \{ \mathcal{N}_j^- \cup \{ v_j \} \} - \{ v_{i_1}, v_{i_2} \}$ (i.e., the $j^{th}$ row of matrix $\mathcal{W}[k_0]$ has exactly $2$ elements equal to $1$ and all the other equal to zero). 
Also, let us assume that $\sum_{v_i \in \mathcal{N}_l^- \cup \{v_l\}} w_{li}[k_0] \leq 1$, $\forall \ v_l \in \mathcal{V} - \{ v_j \}$ (i.e., for every row of $\mathcal{W}[k_0]$ (except row $j$) \textit{at most} one element is equal to $1$ and all the other are equal to zero). 
The above assumptions, regarding matrix $\mathcal{W}$, mean that, during time step $k_0$, only node $v_j$ will receive two mass variables (from nodes $v_{i_1}$ and $v_{i_2}$) and all the other nodes will receive at most one mass variable. 
We have that $z_{j}[k_0 + 1] = z_{i_1}[k_0] + z_{i_2}[k_0]$ and $z_{l}[k_0 + 1] = z_{i}[k_0]$, for $v_l \in \mathcal{V} - \{ v_j \}$ and some $v_i \in \mathcal{V} - \{ v_{i_1}, v_{i_2} \}$ (i.e., node $v_j$ received two nonzero mass variables while all the other nodes received at most one nonzero mass variable, also including its own mass variable). 
Since, at time step $k_0$, we had $p_0$ nodes with nonzero mass variables and at time step $k_0 + 1$ node $v_j$ received (and summed) two nonzero mass variables, while all the other nodes received at most one nonzero mass variable, this means that, at time step $k_0 + 1$, we have $p_0 - 1$ nodes with nonzero mass variables. 
This means that $| \mathcal{V}^+[k_0+1] | < | \mathcal{V}^+[k_0] |$.

%

By extending the above analysis for scenarios where each row of $\mathcal{W}[k]$, at different time steps $k$, has multiple elements equal to $1$ (but $\mathcal{W}[k]$ remains column stochastic) we can see that the number of nodes $v_j$ with nonzero mass variable $z_{j}[k] > 0$ is non-increasing and thus we have $| \mathcal{V}^+[k+1] | \leq | \mathcal{V}^+[k] |$, $\forall \ k \in \mathbb{N}$. 
\end{proof}

\begin{prop}
\label{PROP2_prob}
Consider a strongly connected digraph $\mathcal{G}_d = (\mathcal{V}, \mathcal{E})$ with $n=|\mathcal{V}|$ nodes and $m=|\mathcal{E}|$ edges and $z_j[0] = 1$ and $y_j[0] \in \mathbb{Z}$ for every node $v_j \in \mathcal{V}$ at time step $k=0$. 
Suppose that each node $v_j \in \mathcal{V}$ follows the Initialization and Iteration steps as described in Algorithm~\ref{algorithm_prob}. 
With probability one, we can find $k_0 \in \mathbb{N}$, so that for every $k \geq k_0$ we have 
$$
y^s_j[k] = \sum_{l=1}^{n}{y_l[0]}  \ \ \text{and} \ \ z^s_j[k] = n ,
$$
which means that 
$$
q^s_j[k] = \frac{\sum_{l=1}^{n}{y_l[0]}}{n} ,
$$
for every $v_j \in \mathcal{V}$ (i.e., for $k \geq k_0$ every node $v_j$ has calculated $q$ as the ratio of two integer values). 
\end{prop}

\begin{proof}
From Proposition~\ref{PROP1_prob} we have that $| \mathcal{V}^+[k+1] | \leq | \mathcal{V}^+[k] |$ (i.e., the number of nonzero mass variables is non-increasing). 
We will show that the number of nonzero mass variables is decreasing after a finite number of steps, until, at some $k_0 \in \mathbb{N}$, we have $y_j[k_0] =\sum_{l=1}^{n}{y_l[0]}  \ \ \text{and} \ \ z_j[k_0] = n$, for some node $v_j \in \mathcal{V}$, and $y_i[k_0] = 0  \ \ \text{and} \ \ z_i[k_0] = 0$, for each $ v_i \in \mathcal{V} - \{ v_j \}$). 
In this scenario, (\ref{alpha_z_y}) and (\ref{alpha_q}) hold for each node $v_j$ for the case where $\alpha = 1$. 

The Iteration Steps~$1$ and $2$ of Algorithm~\ref{algorithm_prob}, during time step $k$, can be expressed according to (\ref{y_exch}) and (\ref{z_exch}), where $y[k] = [y_1[k] \ ... \ y_n[k]]^{\rm T}$, $z[k] = [z_1[k] \ ... \ z_n[k]]^{\rm T}$, and $\mathcal{W}[k] = [w_{lj}[k]]$ is an $n \times n$ binary, column stochastic matrix. 
Focusing on (\ref{z_exch}), suppose that, during time step $k_0$, we have $z_i[k_0] > 0$, $z_j[k_0] > 0$ and $w_{li}[k_0] = 1$, $w_{lj}[k_0] = 1$. 
This scenario will occur with probability equal to $(1 + \mathcal{D}_i^+)^{-1}(1 + \mathcal{D}_j^+)^{-1}$ (i.e., as long as nodes $v_i$ and $v_j$ both transmit towards node $v_l$). 
Furthermore, we have that the mass variables of $v_i$ and $v_j$ will not ``merge'' in $v_l$ with probability $1 - (1 + \mathcal{D}_i^+)^{-1}(1 + \mathcal{D}_j^+)^{-1}$.
By extending the above analysis we have that, every $n$ time steps, the probability that two nonzero mass variables ``merge'' is positive and lower bounded by $\big( \prod_{j=1}^{n} (1 + \mathcal{D}_j^+)^{-1} \big)^2$ (i.e., $\text{P}_{\text{merge}} \geq \big( \prod_{j=1}^{n}(1 + \mathcal{D}_j^+)^{-1} \big)^2$).

Thus, from the execution of Algorithm~\ref{algorithm_prob}, we have that the probability that all nonzero mass variables ``merge'' will be arbitrarily close to $1$ for a sufficiently large $k$. 
This means that $\exists k_0 \in \mathbb{N}$ for which $y_j[k_0] =\sum_{l=1}^{n}{y_l[0]},  \ \ \text{and} \ \ z_j[k_0] = n$, for some node $v_j \in \mathcal{V}$, and $y_i[k_0] = 0,  \ \ \text{and} \ \ z_i[k_0] = 0$, for each $ v_i \in \mathcal{V} - \{ v_j \}$. 
Once this ``merging'' of all nonzero mass variables occurs, we have that the nonzero mass variables of node $v_j$ will update the state variables of every node $v_i \in \mathcal{V}$ (because it eventually will be forward to all other nodes) which means that $\exists k_1 \in \mathbb{N}$ (where $k_1 > k_0$) for which $y_i^s[k_1] =\sum_{l=1}^{n}{y_l[0]}  \ \ \text{and} \ \ z_i^s[k_1] = n$, for every node $v_i \in \mathcal{V}$. 
This means that after a finite number of steps, (\ref{alpha_z_y}) and (\ref{alpha_q}) will hold for every node $v_j \in \mathcal{V}$ for the case where $\alpha = 1$. 
\end{proof}

\begin{remark}
It is interesting to note that during the operation of Algorithm~\ref{algorithm_prob}, after a finite number of steps $k_0$, the state variables of each node $v_j \in \mathcal{V}$, become equal to
$y^s_j[k] = \sum_{l=1}^{n}{y_l[0]}$, $z^s_j[k] = n$, so that
$$
q^s_j[k] = \frac{\sum_{l=1}^{n}{y_l[0]}}{n} ,
$$
for $k \geq k_0$. 
This means that (\ref{alpha_z_y}) and (\ref{alpha_q}) will hold for each node $v_j$ for the case where $\alpha = 1$. 
However, this does not necessarily hold for the distributed algorithm presented in the following section. 
\end{remark}

\begin{remark}
It is also worth pointing out that during the operation of Algorithm~\ref{algorithm_prob}, once (\ref{alpha_z_y}) and (\ref{alpha_q}) hold for each node $v_j$ for the case where $\alpha = 1$, then each node also obtains knowledge regarding the total number of nodes in the digraph, since $z^s_j[k] = n$, $\forall v_j \in \mathcal{V}$, which may be useful for determining the number of agents in the network. 
\end{remark}

%
%
%
%
\section{EVENT-TRIGGERED QUANTIZED AVERAGING ALGORITHM}\label{DetAlgorithm}

In this section we propose a distributed algorithm in which the nodes receive quantized messages and perform transmissions according to a set of deterministic \textit{conditions}, so that they reach quantized average consensus on their initial values. The operation of the proposed distributed algorithm is summarized below.

\noindent
\textbf{Initialization:}
Each node $v_j$ assigns to each of its outgoing edges $v_l \in \mathcal{N}^+_j$ a \textit{unique order} $P_{lj}$ in the set $\{0,1,..., \mathcal{D}_j^+ -1\}$, which will be used to transmit messages to its out-neighbors in a round-robin fashion. 
Node $v_j$ has initial value $y_j[0]$ and sets its state variables, for time step $k=0$, as $z_j[0] = 1$, $z^s_j[0] = 1$ and $y^s_j[0] = y_j[0]$, which means that $q^s_j[0] = y_j[0] / 1$.
Then, it chooses an out-neighbor $v_l \in \mathcal{N}_j^+$ (according to the predetermined order $P_{lj}$) and transmits $z_j[0]$ and $y_j[0]$ to that particular neighbor. 
Then, it sets $y_j[0] = 0$ and $z_j[0] = 0$ (since performed a transmission).

\noindent
The iteration involves the following steps:

\noindent
\textbf{Step 1. Receiving:} Each node $v_j$ receives messages $y_i[k]$ and $z_i[k]$ from its in-neighbors $v_i \in \mathcal{N}_j^-$ and sums them along with its stored messages $y_j[k]$ and $z_j[k]$ to obtain
$$
y_j[k+1] = \sum_{v_i \in \mathcal{N}_j^- \cup \{v_j\}} w_{ji}[k]y_i[k] ,
$$
and 
$$
z_j[k+1] = \sum_{v_i \in \mathcal{N}_j^- \cup \{v_j\}} w_{ji}[k]z_i[k] ,
$$
where $w_{ji}[k] = 0$ if no message is received from in-neighbor $v_i \in \mathcal{N}_j^-$; otherwise $w_{ji}[k] = 1$. 

\noindent
\textbf{Step 2. Event-Triggered Conditions:} Node $v_j$ checks the following conditions:
\begin{enumerate}
\item It checks whether $z_j[k+1]$ is greater than $z^s_j[k]$,
\item If $z_j[k+1]$ is equal to $z^s_j[k]$, it checks whether $y_j[k+1]$ is greater than (or equal to) $y^s_j[k]$.
\end{enumerate}
If one of the above two conditions holds, it sets $y^s_j[k + 1] = y_j[k+1]$, $z^s_j[k + 1] = z_j[k+1]$ and $q_j^s[k+1] = \frac{y_j^s[k+1]}{z_j^s[k+1]}$. 

\noindent
\textbf{Step 3. Transmitting:} If the event-trigger conditions above do not hold, no transmission is performed. 
Otherwise, if the event-trigger conditions above hold, node $v_j$ chooses an out-neighbor $v_l \in \mathcal{N}_j^+$ according to the order $P_{lj}$ (in a round-robin fashion) and transmits $z_j[k+1]$ and $y_j[k+1]$. 
Then, since it transmitted its stored mass, it sets $y_j[k+1] = 0$ and $z_j[k+1] = 0$. 
Then, $k$ is set to $k+1$ and the iteration repeats (it goes back to Step~1). 

This event-based quantized mass transfer process is summarized as Algorithm~\ref{algorithm1}, where each node $v_j$ at time step $k$ maintains mass variables $y_j[k]$ and $z_j[k]$ and state variables $y^s_j[k]$ and $z^s_j[k]$ (and $q^s_j[k] = y^s_j[k] / z^s_j[k]$). 
Note that the event trigger conditions effectively imply that no transmission is performed if $z_j[k]=0$.


\noindent
\vspace{-0.5cm}    
\begin{varalgorithm}{2}
\caption{Deterministic Quantized Average Consensus}
\textbf{Input} 
\\ 1) A strongly connected digraph $\mathcal{G}_d = (\mathcal{V}, \mathcal{E})$ with $n=|\mathcal{V}|$ nodes and $m=|\mathcal{E}|$ edges. 
\\ 2) For every $v_j$ we have $y_j[0] \in \mathbb{Z}$. 
\\
\textbf{Initialization} 
\\ Every node $v_j \in \mathcal{V}$: 
\\ 1) Assigns to each of its outgoing edges $v_l \in \mathcal{N}^+_j$ a \textit{unique order} $P_{lj}$ in the set $\{0,1,..., \mathcal{D}_j^+ -1\}$.
\\ 2) Sets $z_j[0] = 1$, $z^s_j[0] = 1$ and $y^s_j[0] = y_j[0]$ (which means that $q^s_j[0] = y_j[0] / 1$). 
\\ 3) Chooses an out-neighbor $v_l \in \mathcal{N}_j^+$ according to the predetermined order $P_{lj}$ (i.e., it chooses $v_l \in \mathcal{N}_j^+$ such that $P_{lj}=0$) and transmits $z_j[0]$ and $y_j[0]$ to this out-neighbor. Then, it sets $y_j[0] = 0$ and $z_j[0] = 0$.
\\
\textbf{Iteration}
\\ For $k=0,1,2,\dots$, each node $v_j \in \mathcal{V}$ does the following:
\\ 1) It receives $y_i[k]$ and $z_i[k]$ from its in-neighbors $v_i \in \mathcal{N}_j^-$ and sets 
$$
y_j[k+1] = \sum_{v_i \in \mathcal{N}_j^- \cup \{v_j\}} w_{ji}[k]y_i[k] ,
$$
and 
$$
z_j[k+1] = \sum_{v_i \in \mathcal{N}_j^- \cup \{v_j\}} w_{ji}[k]z_i[k] ,
$$
where  $w_{ji}[k] = 0$ if no message is received (otherwise $w_{ji}[k] = 1$). 
\\ 2) \underline{Event triggered conditions:} If one of the following two conditions hold, node $v_j$ performs Steps $3$ and $4$ below, otherwise it skips Steps $3$ and $4$.
\\ Condition~$1$: $z_j[k+1] > z^s_j[k]$.
\\ Condition~$2$: $z_j[k+1] = z^s_j[k]$ and $y_j[k+1] \geq y^s_j[k]$.
\\ 3) It sets $z^s_j[k+1] = z_j[k+1]$ and $y^s_j[k+1] = y_j[k+1]$ which implies that 
$$
q^s_j[k+1] = \frac{y^s_j[k+1]}{z^s_j[k+1]} .
$$
4) It chooses an out-neighbor $v_l \in \mathcal{N}_j^+$ according to the order $P_{lj}$ (in a round-robin fashion) and transmits $z_j[k+1]$ and $y_j[k+1]$. Then it sets $y_j[k+1] = 0$ and $z_j[k+1] = 0$.
\\ 5) It repeats (increases $k$ to $k + 1$ and goes back to Step~1).
\label{algorithm1}
\end{varalgorithm}



We now analyze the functionality of the distributed algorithm and we prove that it allows all agents to reach quantized average consensus after a finite number of steps. Depending on the graph structure and the initial mass variables of each node, we have the following two possible scenarios:
\begin{enumerate}
\item[A.] Full Mass Summation (i.e., there exists $k_0 \in \mathbb{N}$ where we have $y_j[k_0] =\sum_{l=1}^{n}{y_l[0]}  \ \ \text{and} \ \ z_j[k_0] = n$, for some node $v_j \in \mathcal{V}$, and $y_i[k_0] = 0  \ \ \text{and} \ \ z_i[k_0] = 0$, for each $ v_i \in \mathcal{V} - \{ v_j \}$). 
In this scenario (\ref{alpha_z_y}) and (\ref{alpha_q}) hold for each node $v_j$ for the case where $\alpha = 1$. 
\item[B.] Partial Mass Summation (i.e., there exists $k_0 \in \mathbb{N}$ so that for every $k \geq k_0$ there exists a set $\mathcal{V}^p[k] \subseteq \mathcal{V}$ in which we have $y_j[k] = y_i[k]$ and $z_j[k] = z_i[k]$, $\forall v_j, v_i \in \mathcal{V}^p[k]$ and $y_l[k] = 0  \ \ \text{and} \ \ z_l[k] = 0$, for each $ v_l \in \mathcal{V} - \mathcal{V}^p[k]$). 
In this scenario (\ref{alpha_z_y}) and (\ref{alpha_q}) hold for each node $v_j$ for the case where $\alpha = | \mathcal{V}^p[k] |$.
\end{enumerate}


An example regarding the scenario of ``Partial Mass Summation'' is given below. 

\begin{example}

Consider a strongly connected digraph $\mathcal{G}_d = (\mathcal{V}, \mathcal{E})$, shown in Fig.~\ref{partial_example}, with $\mathcal{V} = \{ v_1, v_2, v_3, v_4 \}$ and $\mathcal{E} = \{ m_{21}, m_{32}, m_{43}, m_{14} \}$ where each node has an initial quantized value $y_1[0] = 9$, $y_2[0] = 3$, $y_3[0] = 9$ and $y_4[0] = 3$ respectively. 
We have that the average of the initial values of the nodes, is equal to $q = \frac{24}{4}$. 

\begin{figure}[h]
\begin{center}
\includegraphics[width=0.25\columnwidth]{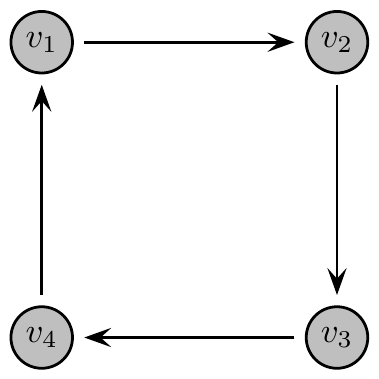}
\caption{Example of digraph for partial mass summation.}
\label{partial_example}
\end{center}
\end{figure}

At time step $k=0$ the initial mass and state variables for nodes $v_1, v_2, v_3, v_4$ are shown in Table~\ref{table_det}. 

\begin{center}
\captionof{table}{Initial Mass and State Variables for Fig.~\ref{partial_example}}
\label{table_det}
{\small 
\begin{tabular}{|c||c|c|c|c|c|}
\hline
Nodes &\multicolumn{5}{c|}{Mass and State Variables for $k=0$}\\
$v_j$ &$y_j[0]$&$z_j[0]$&$y^s_j[0]$&$z^s_j[0]$&$q^s_j[0]$\\
\cline{2-6}
 &  &  &  &  & \\
$v_1$ & 9 & 1 & 9 & 1 & 9 / 1\\
$v_2$ & 3 & 1 & 3 & 1 & 3 / 1\\
$v_3$ & 9 & 1 & 9 & 1 & 9 / 1\\
$v_4$ & 3 & 1 & 3 & 1 & 3 / 1\\
\hline
\end{tabular}
}
\end{center}
\vspace{0.2cm}

\noindent
Then, during time step $k=0$, every node $v_j$ will transmit its mass variables $y_j[0]$ and $z_j[0]$ (since the event-triggered conditions hold for every node). 
The mass and state variables of every node at $k=1$ are shown in Table~\ref{table_det_1}.

\noindent
It is important to notice here that, for time step $k=1$, nodes $v_1$ and $v_3$ have mass variables equal to $y_1[1] = 3$, $z_1[1] = 1$ and $y_3[1] = 3$, $z_3[1] = 1$ but the corresponding state variables are equal to $y^s_1[1] = 9$, $z^s_1[1] = 1$ and $y^s_3[1] = 9$, $z^s_3[1] = 1$. 
This means that at time step $k=1$, the event-triggered conditions do \textit{not} hold for nodes $v_1$ and $v_3$; thus, these nodes will not transmit their mass variables (i.e., they will not execute Steps~$3$ and $4$ of Algorithm~\ref{algorithm1}). 
The mass and state variables of every node at $k=2$ are shown in Table~\ref{table_det_2}.

\begin{center}
\captionof{table}{Mass and State Variables for Fig.~\ref{partial_example} for $k=1$}
\label{table_det_1}
{\small 
\begin{tabular}{|c||c|c|c|c|c|}
\hline
Nodes &\multicolumn{5}{c|}{Mass and State Variables for $k=1$}\\
$v_j$ &$y_j[1]$&$z_j[1]$&$y^s_j[1]$&$z^s_j[1]$&$q^s_j[1]$\\
\cline{2-6}
 &  &  &  &  & \\
$v_1$ & 3 & 1 & 9 & 1 & 9 / 1\\
$v_2$ & 9 & 1 & 9 & 1 & 9 / 1\\
$v_3$ & 3 & 1 & 9 & 1 & 9 / 1\\
$v_4$ & 9 & 1 & 9 & 1 & 9 / 1\\
\hline
\end{tabular}
}
\end{center}
\vspace{0.4cm}

\begin{center}
\captionof{table}{Mass and State Variables for Fig.~\ref{partial_example} for $k=2$} 
\label{table_det_2}
{\small 
\begin{tabular}{|c||c|c|c|c|c|}
\hline
Nodes &\multicolumn{5}{c|}{Mass and State Variables for $k=2$}\\
$v_j$ &$y_j[2]$&$z_j[2]$&$y^s_j[2]$&$z^s_j[2]$&$q^s_j[2]$\\
\cline{2-6}
 &  &  &  &  & \\
$v_1$ & 12 & 2 & 12 & 2 & 12 / 2\\
$v_2$ & 0 & 0 & 9 & 1 & 9 / 1\\
$v_3$ & 12 & 2 & 12 & 2 & 12 / 2\\
$v_4$ & 0 & 0 & 9 & 1 & 9 / 1\\
\hline
\end{tabular}
}
\end{center}
\vspace{0.2cm}

During time step $k=2$ we can see that the event-triggered conditions hold for nodes $v_1$ and $v_3$ which means that they will transmit their mass variables towards nodes $v_2$ and $v_4$ respectively. 
The mass and state variables of every node for $k=3$ are shown in Table~\ref{table_det_3}.

\begin{center}
\captionof{table}{Mass and State Variables for Fig.~\ref{partial_example} for $k=3$}
\label{table_det_3}
{\small 
\begin{tabular}{|c||c|c|c|c|c|}
\hline
Nodes &\multicolumn{5}{c|}{Mass and State Variables for $k=3$}\\
$v_j$ &$y_j[3]$&$z_j[3]$&$y^s_j[3]$&$z^s_j[3]$&$q^s_j[3]$\\
\cline{2-6}
 &  &  &  &  & \\
$v_1$ & 0 & 0 & 12 & 2 & 12 / 2\\
$v_2$ & 12 & 2 & 12 & 2 & 12 / 2\\
$v_3$ & 0 & 0 & 12 & 2 & 12 / 2\\
$v_4$ & 12 & 2 & 12 & 2 & 12 / 2\\
\hline
\end{tabular}
}
\end{center}
\vspace{0.2cm}

Following the algorithm operation we have that, for $k=3$, the event-trigger conditions hold for nodes $v_2$ and $v_4$ which means that they will transmit their masses to nodes $v_1$ and $v_3$ respectively. 
As a result we have, for $k=4$, that the mass variables for nodes $v_1$ and $v_3$ are $y_1[4] = y_4[3] = 12$, $z_1[4] = z_4[3] = 2$ and $y_3[4] = y_2[3] = 12$, $z_3[4] = z_2[3] = 2$ respectively. 
Then, during time step $k=4$, we have that the event-triggered conditions hold for nodes $v_1$ and $v_3$ which means that they will transmit their mass variables to nodes $v_1$ and $v_3$. 
We can easily notice that, during the execution of Algorithm~\ref{algorithm1} for $k \geq 3$, we have $\mathcal{V}^p[k] = \mathcal{V}^p[k + 2]$ (where $\mathcal{V}^p[3] = \{ v_2, v_4 \}$ and $\mathcal{V}^p[4] = \{ v_1, v_3 \}$), which means that the exchange of mass variables between the nodes will follow a \textit{periodic} behavior and the mass variables will never ``merge'' in one node (i.e., $\nexists k_0$ for which $y_j[k_0] =\sum_{l=1}^{n}{y_l[0]}  \ \ \text{and} \ \ z_j[k_0] = n$, for some node $v_j \in \mathcal{V}$, and $y_i[k_0] = 0  \ \ \text{and} \ \ z_i[k_0] = 0$, for each $ v_i \in \mathcal{V} - \{ v_j \}$).

As a result, from Table~\ref{table_det_3}, we can see that for $k \geq 3$ it holds that 
$$
q_j^s[k] = q = \frac{24 / \alpha}{4 / \alpha} ,
$$
for every $v_j \in \mathcal{V}$, for $\alpha = | \mathcal{V}^p[k] | = 2$. 
This means that, after a finite number of steps, every node $v_j$ will obtain a quantized fraction $q_j^s$ which is equal to the average $q$ of the initial values of the nodes. 
\end{example}

\begin{remark}
Note that the periodic behavior in the above graph is not only a function of the graph structure but also of the initial conditions. Also note that, in general, the priorities will also play a role because they determine the order in which nodes transmit to their out-neighbors (in the example, priorities do not come into play because each node has exactly one out-neighbor).
\end{remark}

\begin{prop}
\label{PROP1_det}
Consider a strongly connected digraph $\mathcal{G}_d = (\mathcal{V}, \mathcal{E})$ with $n=|\mathcal{V}|$ nodes and $m=|\mathcal{E}|$ edges. 
The execution of Algorithm~\ref{algorithm1} will allow each node $v_j \in \mathcal{V}$ to reach quantized average consensus after a finite number of steps, bounded by $n^5$. 
\end{prop}

\section{SIMULATION RESULTS} \label{results}

In this section, we present simulation results and comparisons. 
Specifically, we present simulation results of the proposed distributed algorithms for the digraph $\mathcal{G}_d = (\mathcal{V}, \mathcal{E})$ (borrowed from \cite{2014:RikosHadj}), shown in Fig.~\ref{simul}, with $\mathcal{V} = \{ v_1, v_2, v_3, v_4, v_5, v_6, v_7 \}$ and $\mathcal{E} = \{ m_{21}, m_{51}, m_{12}, m_{52}, m_{13}, m_{53}, m_{24}, m_{54}, m_{65}, m_{75},$ $m_{36}, m_{47}, m_{67} \}$, where each node has initial quantized values $y_1[0] = 5$, $y_2[0] = 4$, $y_3[0] = 8$, $y_4[0] = 3$, $y_5[0] = 5$, $y_6[0] = 2$, and $y_7[0] = 7$, respectively. 
The average $q$ of the initial values of the nodes, is equal to $q = \frac{34}{7}$. 


\begin{figure}[h]
\begin{center}
\includegraphics[width=0.40\columnwidth]{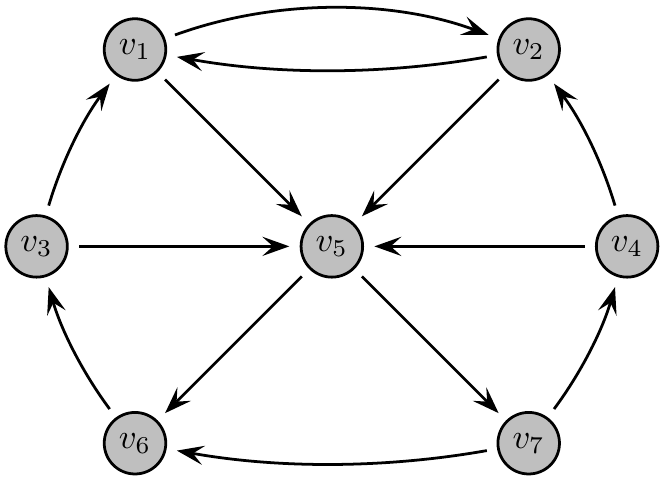}
\caption{Example of digraph for comparison of Algorithms~\ref{algorithm_prob} and \ref{algorithm1}.}
\label{simul}
\end{center}
\end{figure}


In Figure~\ref{simul_plot} we plot the state variable $q_j^s[k]$ of every node $v_j \in \mathcal{V}$ as a function of the number of iterations $k$ for the digraph shown in Fig.~\ref{simul}. 
The plot demonstrates that the proposed distributed algorithms are able to achieve a common quantized consensus value to the average of the initial states after a finite number of iterations.

\begin{figure} [ht]
\centering
\includegraphics[width=60mm]{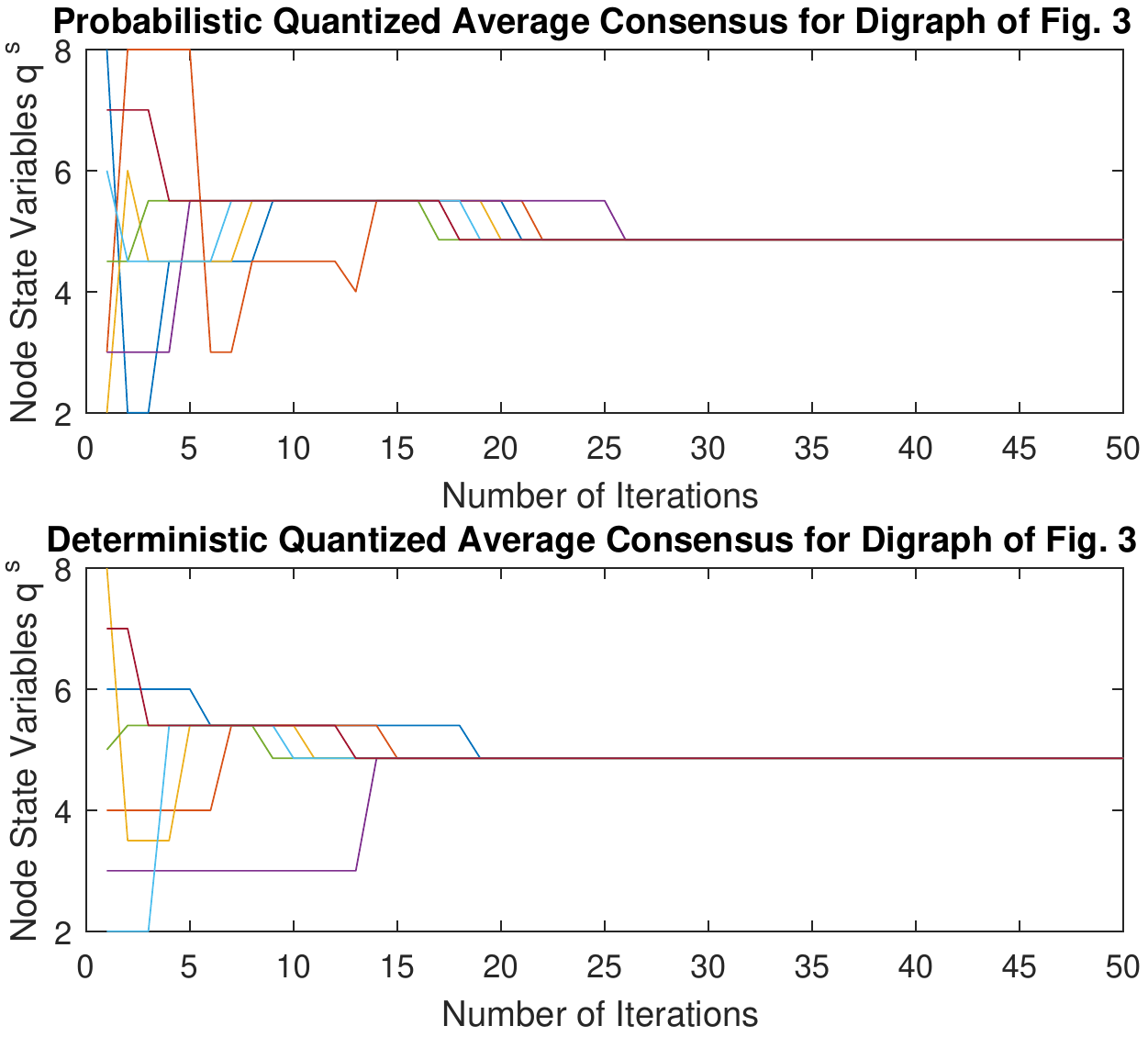}
\caption{Comparison between Algorithm~\ref{algorithm_prob} and Algorithm~\ref{algorithm1} for the digraph shown in Fig.~\ref{simul}. 
\emph{Top figure:} Node state variables plotted against the number of iterations for Algorithm~\ref{algorithm_prob}. 
\emph{Bottom figure:} Node state variables plotted against the number of iterations for Algorithm~\ref{algorithm1}\vspace{-0.45cm}.}
\label{simul_plot}
\end{figure}

\section{CONCLUSIONS}\label{future}

We have considered the quantized average consensus problem and presented one randomized and one deterministic distributed averaging algorithm in which the processing, storing and exchange of information between neighboring agents is subject to uniform quantization. 
We analyzed the operation of the proposed algorithms and established that they will reach quantized consensus after a finite number of iterations.

In the future we plan to investigate the dependence of the graph structure with full and partial mass summation of the initial values. 
Furthermore, we plan to extend the operation of the proposed algorithm to more realistic cases, such as transmission delays over the communication links and the presence of unreliable links over the communication network. 

\balance

\vspace{-0.3cm}

\bibliographystyle{IEEEtran}


\balance

\end{document}